%% file: generalization.tex
\documentclass{article}
\usepackage[utf8]{inputenc}
\usepackage{amsmath}
\usepackage{amsfonts}
\usepackage{xcolor}
\usepackage{subcaption}
\usepackage{float}
\usepackage[top=1in, bottom=1in, left=1in, right=1in]{geometry}
\usepackage{amsthm}
\usepackage{mathtools}
\usepackage{graphicx}
\usepackage{bbm}

\title{Efficient Near-Optimal Codes for General Repeat Channels}
\author{Francisco Pernice\thanks{Departments of Mathematics and Computer Science, Stanford University. fpernice@stanford.edu. Research supported by CURIS 2021.} , Ray Li\thanks{Department of Computer Science, Stanford University. rayyli@cs.stanford.edu. Research supported by NSF Grants DGE-1656518, CCF-1814629} , Mary Wootters\thanks{Departments of Computer Science and Electrical Engineering, Stanford University.  marykw@stanford.edu.  Research partially supported by NSF Grant CCF-1844628 and by a Sloan Research Fellowship.}}
\date{\today}

\usepackage{hyperref}
\usepackage[
backend=biber,
style=alphabetic,
sorting=ynt,
maxbibnames=99
]{biblatex}
\newcommand{\E}{\mathbb{E}}

\DeclareMathOperator{\BDC}{\mathrm{BDC}}
\DeclareMathOperator{\TDC}{\mathrm{TDC}}

\DeclareMathOperator{\TRIM}{\mathrm{TRIM}}

\DeclareMathOperator{\RC}{\mathrm{RC}}
\DeclareMathOperator{\TRC}{\mathrm{TRC}}

\DeclareMathOperator{\DC}{\mathrm{DC}}

\DeclareMathOperator{\PRC}{\mathrm{PRC}}

\newcommand{\eps}{\varepsilon}
\newcommand{\CC}{\mathcal{C}}
\newcommand{\DD}{\mathcal{D}}

\newcommand{\N}{\mathbb{N}}
\newcommand{\1}{\mathbbm{1}}
\newcommand{\pr}{\mathbb{P}}

\newcommand{\TT}{\mathcal{T}}

\newcommand{\II}{\mathcal{I}}

\newcommand{\R}{\mathbb{R}}

\newcommand{\til}[1]{\widetilde{#1}}
\newcommand{\hatt}[1]{\widehat{#1}}

\newcommand{\eqdist}{\stackrel{\DD}{=}}

\DeclareMathOperator{\Ch}{\mathsf{Ch}}
\DeclareMathOperator{\Enc}{\mathsf{Enc}}
\DeclareMathOperator{\poly}{\mathsf{poly}}

\DeclareMathOperator{\Dec}{\mathsf{Dec}}

\theoremstyle{definition}

\newtheorem{thm}{Theorem}
\newtheorem{lem}[thm]{Lemma}

\newtheorem{claim}[thm]{Claim}

\newtheorem{prop}[thm]{Proposition}

\newtheorem{defn}[thm]{Definition}

\numberwithin{thm}{section}

\begin{document}

\maketitle

\begin{abstract}
Given a probability distribution $\DD$ over the non-negative integers, a \emph{$\DD$-repeat channel} acts on an input symbol by repeating it a number of times distributed as $\DD$. For example, the binary deletion channel ($\DD=Bernoulli$) and the Poisson repeat channel ($\DD=Poisson$) are special cases. We say a $\DD$-repeat channel is \emph{square-integrable} if $\DD$ has finite first and second moments. In this paper, we construct explicit codes for all square-integrable $\DD$-repeat channels with rate arbitrarily close to the capacity, that are encodable and decodable in linear and quasi-linear time, respectively. We also consider possible extensions to the repeat channel model, and illustrate how our construction can be extended to an even broader class of channels capturing insertions, deletions, and substitutions.

Our work offers an alternative, simplified, and more general construction to the recent work of Rubinstein \cite{rubin}, who attains similar results to ours in the cases of the deletion channel and the Poisson repeat channel. It also slightly improves the runtime and decoding failure probability of the polar codes constructions of Tal et al. \cite{polar} and of Pfister and Tal \cite{polar2} for the deletion channel and certain insertion/deletion/substitution channels. Our techniques follow closely the approaches of Guruswami and Li \cite{explicit-codes1} and Con and Shpilka \cite{con-shpilka}; what sets apart our work is
that to obtain our result, we show that a capacity-achieving code for the channels in question can be assumed to have an ``approximate balance'' in the frequency of zeros and ones of all sufficiently long substrings of all codewords. This allows us to attain near-capacity-achieving codes
in a
general setting. We consider this ``approximate balance'' result to be of independent interest, as it can be cast in much greater generality than just repeat channels.

\end{abstract}

\section{Introduction}
Fixing a probability distribution $\DD$ over the natural numbers $\N,$ a $\DD$-repeat channel acts on an input bit by repeating it a number of times distributed like $\DD.$ Special cases include the binary deletion channel, Poisson repeat channel, and the sticky channels (the latter two were introduced by Mitzenmacher et al \cite{simple-lower-bound-1/9, sticky-channels}). 
We say a $\DD$-repeat channel is \emph{square-integrable} if $\DD$ has finite first and second moments.
In general, the output of a $\DD$-repeat channel has random length, and does not preserve synchronization; that is, one cannot see the index at the input of a given observed bit at the output. This introduces \emph{memory} into the channel, making its analysis much more complicated than its memoryless counterparts. 
For example, in stark contrast with the memoryless case, even in the simplest case of the binary deletion channel (where $\DD=Bernoulli(p)$), the capacity is unknown, although several lower and upper bounds have been proved (see \cite{mahdi-review, mit-review} for two excellent surveys on synchronization channels). 

More recently, progress has been made on constructing \emph{explicit} and \emph{efficient} codes whose rates approximate the state of the art lower bounds on capacity for certain simple special cases of repeat channels. Guruswami and Li \cite{explicit-codes1} gave the first explicit and efficient codes for the deletion channel with $\Theta(1-d)$ rate, achieving a rate of $(1-d)/120.$ This was later improved by Con and Shpilka \cite{con-shpilka} to $(1-d)/16.$ Tal et al. \cite{polar} gave a construction using polar codes proved to achieve the capacity of the deletion channel by considering a sequence of hidden-markov input processes that approach the maximum mutual information. In \cite{polar2}, their construction was extended to a more general model of synchronization errors, which allows for simple insertions and bit flips.  Very recently, Rubinstein \cite{rubin} gave a black-box construction, which takes a general (inefficient and non-explicit) code for the deletion channel or the Poisson repeat channel of a given rate $R$ and produces an efficient and explicit code of rate $R-\eps$, for any $\eps>0.$ In particular, this yields an efficient and explicit code achieving capacity on these channels.  However, to our knowledge, no efficient and explicit code construction has been given of even non-trivial rate for general square-integrable repeat channels.

In this paper, we show that, by extending the techniques of \cite{explicit-codes1, con-shpilka}, we can obtain codes for \emph{any} square-integrable repeat channel that are efficiently encodable and decodable and of rate within $\eps$ of the capacity, for any $\eps>0.$ We also illustrate how our construction can give explicit, efficient capacity achieving codes for an even broader class of channels capturing insertions, deletions, and substitutions.

As mentioned above, similar results appeared in the literature before, and our result differs in the following ways.  First, the work \cite{rubin} proves the same result for the deletion channel and the Poisson repeat channel. Our construction generalizes the result to general repeat channels, and we illustrate how it can be generalized further to channels capturing insertions, deletions, and substitutions. We also believe our proof is simpler. Second, the works \cite{polar,polar2} obtain similar results for the deletion channel \cite{polar} and insertion/deletion/substitution channels \cite{polar2}. Compared to these works, we give slightly faster decoding algorithms and slightly smaller error probability: for any $0<\nu'<\nu<1/3$, \cite{polar,polar2} give decoding error probability $e^{-\Omega(n^{\nu'})}$ in time $O(n^{1+3\nu})$, while we achieve decoding error probability $e^{-\Omega(n)}$ in time $O(n\poly\log n)$.

The simplicity and generality of our construction arises from the observation that, in a broad class of channels, there exist capacity achieving codes all of whose codewords are ``locally approximately balanced" in zeros and ones (Proposition~\ref{zeros-distribution-lemma}).
This balance property allows us to directly use code concatenation, in contrast to \cite{rubin} which needs recursive concatenation.
Our construction concatenates one such ``locally approximately balanced" constant length inner codes with an outer code correcting worst-case insertions and deletions.

\subsection{Organization}
In Section~\ref{preliminaries} we review some background material needed for the proofs. In Section \ref{sec:main-result} we give the construction, compare it with the recent work of \cite{rubin}, and prove its correctness. Finally in Section \ref{extensions} we give a brief overview of how our results can be extended to a more general error model allowing for insertion and substitution errors.

\section{Preliminaries}\label{preliminaries}
\subsection{Notation}
In what follows, $\{0,1\}^n$ for $n \in \N \cup\{\infty\}$ ($\N = \{0,1,2,\dots\}$) denotes the set of bit strings of length $n$; we also let $\{0,1\}^* = \bigcup_{n \in \N} \{0,1\}^n.$ For $x\in \{0,1\}^n$, we let $x_j^k$ denote the substring of $x$ starting at index $j$ and ending at $k,$ inclusive, and unless specified otherwise, we let $x_i := x_i^i$. For $n\in \N,$ we let $[n] = \{1,2,\dots, n\};$ even if $n\in \R_+,$ we let $[n]:=[\lfloor n\rfloor].$  For two strings $x,y\in \{0,1\}^*,$ we let $xy$ denote their concatenation, and for $k\in \N$, $(x)^k$ denotes the $k$-wise concatenation of $x$ with itself; we let $(x)^0$ be the empty string. For $x\in \{0,1\}^n,$ $|x|=n$ denotes the length of $x.$ We denote the capacity of an arbitrary channel $\Ch$ (which we define below) by $Cap(\Ch).$ All logs (hence entropies, etc.) in this paper are base 2. For a probability distribution $\DD$ over $\R,$ we let $\mu(\DD)$ denote the expectation; whenever $\DD$ is understood from context we sometimes just write $\mu$. Similarly we let $\sigma^2(\DD)$ denote the variance, for which we sometimes just write $\sigma^2$.
Throughout the paper, by ``quasi-linear time'' we mean time $O(n\poly(\log n))$. 
\subsection{Basic Concepts}

For completeness, we give a definition of a general binary communication channel, introducing the notation we will use in this paper.
\begin{defn}
For $\Omega$ a probability space, a binary communication channel is a map $\Ch:\Omega\times \{0,1\}^*\to\{0,1\}^*.$ For $x\in \{0,1\}^*,$ we write $\Ch x$ to denote the random variable $\omega\mapsto \Ch(\omega, x).$
\end{defn}

In this paper we deal specifically with \emph{square-integrable binary repeat channels}, which we define next.
\begin{defn}\label{D-repeat-channels}
For a probability distribution $\DD$ over $\N$, let $\Omega = \N^\infty$ (the infinite product space), with a $\DD^\infty$ measure (the infinite product measure). The \emph{binary $\DD$-repeat channel} is defined as $\RC_\DD(\omega, x) = (x_1)^{\omega_1}(x_2)^{\omega_2}\dots (x_n)^{\omega_n}$ for an input $x\in \{0,1\}^n.$ We say $\RC_\DD$ is \emph{square-integrable} if $\mu(\DD)<\infty$ and $\sigma^2(\DD)<\infty.$ 
\end{defn}

That is, each bit sent through the $\RC_\DD$ gets repeated $R\sim \DD$ times. 
 We note that it is well-defined to speak of the index at the input that gave rise to a given bit at the output: the origin bit of the $j$th bit at the output is the $\min\{i\geq 1:\sum_{k=1}^i \omega_k \geq j\}$'th bit at the input.

We now define error correcting codes and channel capacity.
\begin{defn}
A \emph{code} is a family $\CC = \{\CC_n\}_{n\geq 1}$ of subsets $\CC_n\subseteq \{0,1\}^n,$ and its \emph{rate} is defined as $R_n = \frac{1}{n}\log |\CC_n|$; we also sometimes refer to $R=R(\CC) = \lim_{n\to\infty} R_n$ as the rate. An \emph{encoding algorithm} for $\CC$ is a family of maps $\Enc:\{0,1\}^{nR_n }\to \CC_n$ and a \emph{decoding algorithm} is a family of maps $\Dec:\{0,1\}^* \to \{0,1\}^{nR_n }$ (we drop $n$ from the notation). We say that the code $\CC$ is \emph{sound} with respect to a channel $\Ch$ if there exist encoding and decoding algorithms $\Enc$ and $\Dec,$ respectively, for $\CC$ such that
\[
\pr(\Dec(\Ch \Enc(X)) \neq X)\to 0
\]
as $n\to\infty,$ where $X$ is taken uniformly from $\{0,1\}^{nR_n}.$
\end{defn}
\begin{defn}
Fix a channel $\Ch$ and let $\mathcal{S}$ be the collection of all codes that are sound with respect to $\Ch.$ The \emph{capacity} of $\Ch$ is defined as 
\[
Cap(\Ch) = \sup_{\CC\in \mathcal{S}} R(\CC).
\]
\end{defn}

Finally we define the \emph{trimming repeat channels}, which unlike the objects defined above are non-standard, but which will be an important part of our construction. We note that ``trimming versions'' of synchronization channels appear in \cite{polar, polar2} and that there they play a role similar to the one in our construction.
\begin{defn}\label{trim-bdc}
Let $\TRIM$ be a (deterministic) channel which acts on $x \in \{0,1\}^n$ by deleting the longest possible substrings at the beginning and end of $x$ consisting entirely of zeros. 

Specifically, $$\TRIM x = x_{\min\{i\in [n]: x_i = 1\}}^{\max\{i\in [n]: x_i = 1\}}$$ (or the empty string if $x$ is all zeros). 
Let $\RC_\DD$ be as in Definition \ref{D-repeat-channels}. We then define the \emph{trimming $\DD$-repeat channel} by the composition $\TRC_\DD := \TRIM \circ \RC_\DD.$
\end{defn}

\subsection{Generalized Shannon's Theorem}
For the square-integrable $\DD$-repeat channels, as well as a wide class of other synchronization channels, Dobrushin \cite{generalized-shannons-thm} showed that the capacity is given by a certain limit of the finite-length message maximum mutual information between input and output; this extended the fundamental result of Shannon \cite{shannon} for memoryless channels. Here we state his theorem in our setting and notation. We refer the reader to the excellent survey of Cheraghchi and Ribeiro \cite{mahdi-review} for an illuminating discussion. Before the theorem we give a simple (non-general) definition of a stationary ergodic process, which will be important in our proof.
\begin{defn}
A stochastic process $\{X_j\}_{j\geq 1}$ is \emph{stationary} if for every $j,N\in \N$ we have $(X_1,\dots,X_N)\eqdist (X_{j+1},\dots,X_{j+N})$, where $\eqdist$ denotes equality in distribution. Moreover, the process $\{X_j\}_{j\geq 1}$ is \emph{stationary ergodic} if it is stationary and it satisfies Birkhoff's Pointwise Ergodic Theorem, i.e. for every $f\in L^1$ we almost surely have
\[
\E f(X_1) = \lim_{n\to\infty}\frac{1}{n}\sum_{j=1}^n f(X_j).
\]
\end{defn}
\begin{thm}[Generalized Shannon's Theorem \cite{generalized-shannons-thm}]\label{shannon-thm}
Consider a channel $\Ch$ that acts independently on each input bit, and appends the corresponding outputs, i.e. we have $\Ch x \eqdist (\Ch x_1)(\Ch x_2)\dots(\Ch x_n)$ for $x\in \{0,1\}^n.$ Suppose further that $\E |\Ch b| < \infty$ for $b\in \{0,1\}$, i.e. the channel output has finite expected length for each input bit. Then the capacity is given by
\[
Cap(\Ch) = \lim_{n\to\infty} \frac{1}{n} \sup_{X^n} I(X^n;Y^n),
\]
where the sup is taken over all random variables $X^n$ supported on $\{0,1\}^n$ and $Y^n = \Ch X^n$. Moreover, the capacity is achieved by a stationary ergodic input process.
\end{thm}
We remark that Theorem \ref{shannon-thm} evidently applies to square-integrable repeat channels. Whenever it's understood from context, we will drop the parameter $n$ and just write $X$ for a random variable supported on $\{0,1\}^n,$ and let $Y = \Ch X.$ 
For channels for which the Generalized Shannon's Theorem doesn't necessarily apply (like trimming $\DD$-repeat channels), we refer to the limit $\lim_{n\to\infty} \frac{1}{n} \sup_{X^n} I(X^n;Y^n)$ as the \emph{information rate} of the channel; hence Theorem \ref{shannon-thm} shows that for the channels to which it applies, the capacity and the information rate coincide.
We emphasize that the fact that the capacity is attained by a stationary ergodic process in Theorem \ref{shannon-thm} will be instrumental in our construction.

\subsection{Worst-case insertion/deletion codes}

We now state the result of Haeupler and Shahrasbi \cite{adversarial-codes} which will be our \emph{outer code}, as it was in the works of Guruswami and Li \cite{explicit-codes1} and Con and Shpilka \cite{con-shpilka}. We first need a definition.
\begin{defn}
For $x \in \{0,1\}^n$ and $b\in \{0,1\},$ an \emph{insertion} and a \emph{deletion} at some index $j\in [n]$ are the transformations $x\mapsto x_1^jbx_{j+1}^n$ and $x\mapsto x_1^{j-1}x_{j+1}^n$, respectively.
The \emph{edit distance} between two strings $x,y$ is the minimum number of insertions and/or deletions to convert $x$ into $y.$

\end{defn}
We can now state the result.
\begin{thm}[\cite{adversarial-codes}, \cite{adversarial-codes-improved}]\label{adversarial-codes-thm}
For every $\eps,\delta \in (0,1)$ there exists a family of codes $\CC_n$ of rate $1-\delta-\eps$ over an alphabet $\Sigma$ of size $O_\eps(1)$ that can (deterministically) correct insertion/deletion (worst-case) errors resulting in an edit distance at most $\delta n$. Moreover, the $\CC_n$ have encoding and decoding algorithms that run in linear and quasi-linear time, respectively.
\end{thm}

\section{Main Result}\label{sec:main-result}
Our main result is the existence of efficient near-optimal codes for square-integrable repeat channels with rates approaching capacity. When restricted to the binary deletion channel or the Possion repeat channel, our construction streamlines the approach of \cite{rubin}. Specifically, we prove the following:
\begin{thm}\label{main-thm}
Fix a square-integrable repeat channel $\RC_\DD$. For every $\eps>0$, there exists a sound code $\CC$ with rate $R$ for the $\RC_\DD$ with $R \geq Cap(\RC_\DD)-\eps$ and linear and quasi-linear time encoding and decoding algorithms, respectively.
Moreover, the decoder has probability of failure $e^{-\Omega(n)}$.
\end{thm}
We organize the remaining of this section as follows. In Section \ref{construction} we give the construction, and comment on how it achieves the promised rate and runtime guarantees. Then in Section \ref{comparison-prior-work} we compare our construction to that of Rubinstein \cite{rubin}, explaining the differences. Finally in Section \ref{main-proof} we give the bulk of the proof of correctness of our construction.

\subsection{Construction}\label{construction}

We prove in Lemma \ref{trim-bdc-cap-lemma} that the information rates of the $\RC_\DD$ and $\TRC_\DD$ are the same. In Proposition~\ref{zeros-distribution-lemma}, we further show that we can assume the existance of a general (non-explicit and inefficient) code $\CC_{in}$ for the $\TRC_\DD$ such that each sufficiently long substring of each codeword in $\CC_{in}$ is approximately balanced in zeros and ones (see Proposition~\ref{zeros-distribution-lemma}), with rate $R \geq Cap(\RC_\DD)-\eps$, for any $\eps>0$.  This will be the \emph{inner code} in our construction, which we assume has (not necessarily efficient) encoding and decoding algorithms $\Enc_{in}$ and $\Dec_{in}$, respectively. Then, as in the work of Con and Shpilka \cite{con-shpilka}, for a codeword length $m$ to be fixed later, we take $2^m$ as the desired alphabet size for the \cite{adversarial-codes, adversarial-codes-improved} code (i.e. $|\Sigma| = 2^m$ in Theorem \ref{adversarial-codes-thm}), making sure to take $m$ large enough for the code of \cite{adversarial-codes, adversarial-codes-improved} to be effective.

Our encoding procedure $\Enc : \{0,1\}^{km}\to \{0,1\}^n$ for some $x\in \{0,1\}^{km}$ works as follows:
\begin{enumerate}
    \item We split $x$ into $x_1,\dots,x_k$, with $|x_j|=m,$ and we view each $x_j$ as a member of $\Sigma,$ hence $x \in \Sigma^k.$ We then use the encoder of Theorem \ref{adversarial-codes-thm} (call it $\Enc_{out}$) to encode $x.$ This yields $\til{x} = \Enc_{out}(x) \in \Sigma^{k/(1-\delta-\eps)}.$
    \item We again split $\til{x}$ into $\til{x}_1,\dots,\til{x}_{k'}$ where $\til{x}_j \in \Sigma, k' = k/(1-\delta-\eps),$ and view each $\til{x}_j$ as an element in $\{0,1\}^m.$ We then encode each $\til{x}_j$ with our inner code to produce $\hatt{x}_j = \Enc_{in}(\til{x}_j) \in \{0,1\}^{m/(R-\eps)},$ where, by taking $m$ large enough, we have made the rate of the inner code $R-\eps.$ We note that since $m = O(1),$ this encoding process is done in $O(1)$ time.
    \item Finally we concatenate the $\hatt{x}_j$ and put buffers of all zeros in between. Specifically, our final encoding of $x$ is 
    \[
    \Enc(x) = \hatt{x}_1 0^b \hatt{x}_2 0^b\dots 0^b \hatt{x}_{k'},
    \]
    where $b = b(m) = \eta m$ is a constant independent of $n= k'\cdot (\frac{m}{R-\eps} + b) = km/(Cap(\RC_\DD) - \psi(\eps,\delta,\eta, k,m))$ with $\psi\to 0$ as $\eps,\delta,\eta\to 0$ and $k,m\to\infty$, so by taking $\eps,\delta,\eta$ small enough and $m$ large enough, we can make the rate of the code get arbitrarily close to $Cap(\RC_\DD)$.
\end{enumerate}
We note that since $\Enc_{out}$ runs in linear time, so does our encoding $\Enc$. For the decoding $\Dec$ of a received string $y \in \{0,1\}^*$, we reverse the steps above:
\begin{enumerate}
    \item We identify the buffers of zeros by interpreting any maximal contiguous block of $\geq \frac{\mu}{2} \eta m $ zeros as a buffer. We remove the buffers, producing the received inner strings $y_1,\dots,y_\ell.$ 
    \item We decode each $y_j$ with our inner code to produce $\til{y}_j =\Dec_{in}(y_j)\in \{0,1\}^m$ for $j\leq \ell.$ 
    \item We interpret each $\til{y}_j$ as a letter in $\Sigma,$ and we decode the concatenation $\til{y}=\til{y}_1\dots\til{y}_\ell \in \Sigma^\ell$ with the outer code, to produce our final decoding of $y$:
    \[
    \Dec(y) = \Dec_{out}(\til{y}) \in \{0,1\}^{km}.
    \]
\end{enumerate}
We note that the identification of the buffers runs in linear time and $\Dec_{out}$ runs in quasi-linear time, hence our overall decoding $\Dec$ runs in quasi-linear time as well.

\subsection{Technical Comparison With \cite{rubin}}\label{comparison-prior-work}
As mentioned above, Theorem~\ref{main-thm} is similar to the result of Rubinstein in \cite{rubin}, and there are also similarities in the construction.  Thus, before we prove Theorem~\ref{main-thm}, we outline the ways in which our construction differs from the one in \cite{rubin}, and mention how this allows us to obtain more general results with a simpler proof.

\paragraph{Exploitation of Code Structure.}  In order for a construction based on concatenation and buffers to work, two hurdles need to be overcome: (1) the trimming that results from the separation of buffers from codewords may yield the decoding algorithm ineffective, and (2) distinguishing the buffers from the codewords in the first place might be difficult. We overcome these hurdles by showing that we may assume without loss of generality that the inner code we start with has some \emph{structure}; i.e. we can assume the code in fact codes for the $\TRC_\DD$ and has an approximately balanced distribution of zeros and ones, making the buffers easy to spot. By contrast, Rubinstein's construction \cite{rubin} is able to take any code for the binary deletion channel $\BDC_d$ or \emph{Poisson repeat channel} $\PRC_\lambda$ and use it unchanged as an inner code, and thus does not exploit any structural properties of the inner code. This is achieved in \cite{rubin} by a very careful construction of the buffers (outlined below), which allows one to correctly identify the buffers and remove them while trimming only a very small number of bits with high probability. Because so few bits are trimmed, a further careful analysis shows that these trimmed bits can only decrease the probability of error of a given decoding algorithm by a small amount, hence maintaining the asymptotic soundness property. Because these arguments are delicate, they require a longer and more complicated proof than the one we offer.

\paragraph{Generalization to Other Channels.} 
As our analysis relies on assuming without loss of generality a particular code structure, it is natural to ask how generally this structure can be assumed.  
As we show in the next section, the structure essentially follows from the generalized Shannon's Theorem \ref{shannon-thm}, and hence applies far beyond the $\BDC$ or $\PRC_\lambda.$ This allows us to easily extend our construction to the more general setting of square-integrable repeat channels. Further, as we outline in Section~\ref{sec:general}, our construction can be made to apply in an even more general setting (although beyond repeat channels, the proofs become more complicated). 
In contrast, extending the more delicate analysis of \cite{rubin} beyond the $\BDC$ and $\PRC_\lambda$ seems (to us) more challenging. 

\paragraph{Construction of Buffers.} In our construction, like the ones of Guruswami and Li \cite{explicit-codes1} and Con and Shpilka \cite{con-shpilka}, the buffers are simple long sequences of zeros that are identified by noticing their unusually high density of zeros as compared to codewords. Rubinstein on the other hand constructs buffers of the form $0^{b/2}1^{b/2}$, i.e. a sequence of all-zeros followed by a sequence of all-ones. The benefit of these ``valley'' buffers is that they have an easily identifiable center point at the receiver: the point of transition from zero to one. Hence, in a sequential decoding of the inner codewords, one can iteratively produce prior estimates of where the received buffer's center point should be, and iteratively align the estimate as long as it lands \emph{anywhere} in the buffer. Hence the error does not compound, and very high-accuracy identification of buffers is obtained.

\paragraph{Outer Code Choice.} Unlike our construction, the construction in \cite{rubin} is recursive: a given message to encode of length $m$ is divided into chunks of length $\sqrt{m},$ which are recursively encoded and then treated as the letters of the outer code. This recursive construction is able to preserve perfect synchronization with high probability. Hence, the outer code used is a Reed-Solomon code, which corrects replacement-type errors. This is in contrast with our construction and those of \cite{explicit-codes1, con-shpilka}, which need the outer code to correct a $\delta$-fraction of insertions/deletions.

\paragraph{Explicitness.} The construction of \cite{rubin} is slightly more explicit than ours in the following sense. In our construction, the only object that is not explicitly given is the inner code, which has $O(1)$ blocklength, so we can enumerate all possible codes of that blocklength until finding the desired one in $O(1)$ time. Notably, it is \emph{not} the case that this exhaustive search would be saved if one started out with an explicitly constructed but not necessarily efficiently encodable/decodable code for a $\RC_\DD$. Indeed, we further require this code to be robust to trimming errors and to have an approximately balanced distribution of zeros and ones, properties that need not be satisfied by general $\RC_\DD$ codes. By contrast, Rubinstein's construction allows for direct use of any code for the $\BDC$ or $\PRC_\lambda$ as the inner code, and thus is explicit in a stronger sense if the inner code one starts with is given explicitly.

\subsection{Proof of Correctness}\label{main-proof}
We organize the proof of Theorem \ref{main-thm} as follows. First we prove that the capacities of the repeat channels are unchanged if we trim off the zeros at the ends of the output. Second, we argue that we can assume there exist capacity-achieving codes with a sufficiently balanced distribution of zeros and ones in all its codewords. Finally, we put these results together into our proof of correctness of the construction given in Section \ref{construction}.

We begin with the first required result. 
\begin{lem}\label{trim-bdc-cap-lemma}
Let $\RC_\DD$ be a square-integrable repeat channel. Then the information rate for the $\TRC_\DD$ is $Cap(\RC_\DD).$
\end{lem}
\begin{proof}
Let $X$ be supported on $\{0,1\}^n$ and $Y = \RC_\DD X.$ Let $L = \ell(Y),R=r(Y)$ be as in Definition \ref{trim-bdc} for the random string $Y:$ they are the (random) indices that mark the all-zero substrings that would be trimmed if $Y$ were passed through $\TRIM.$ Now let $\til{L},\til{R}$ be the indices of the bits in $X$ that, when $X$ is passed through the $\RC_\DD,$ end up at indices $L,R$ in $Y.$ We will prove the following claim.
\begin{claim}\label{cl}
Let $Y=\RC_\DD X$ and $Y' = \TRC_\DD X.$ 
Then
\begin{equation}\label{eq:cl1}|I(X;Y) - I(X;Y|\til{L},\til{R})| = o(n)\end{equation} 
and
\begin{equation}\label{eq:cl2} \lim_{n\to\infty} \frac{1}{n}\sup_X I(X;Y|\til{L},\til{R}) = \lim_{n\to\infty} \frac{1}{n}\sup_{X} I(X;Y').
\end{equation}

\end{claim}
By the generalized Shannon's Theorem, Claim~\ref{cl} proves the lemma. We will establish the claim by assuming that the support of $\DD$ is bounded, with a bound linear in the block length $n,$ i.e. there exists deterministic $B>0$ such that if $R\sim \DD,$ then $R\leq Bn$ with probability 1. This assumption is without loss of generality: if $\DD$ has unbounded support, we consider the ``truncation at $Bn$'' $\DD_n$ such that if $R\sim \DD, R'\sim \DD_n$ we have
\[
\pr(R' = k) = 
\begin{cases}
\pr(R = k) / (\sum_{\ell \leq Bn} \pr(R = \ell)) &\text{if }k\leq Bn\\
0&\text{otherwise}.
\end{cases}
\]
For each $x\in \{0,1\}^n$, we have $\TRC_{\DD_n}x\eqdist \TRC_\DD x$ and $\RC_{\DD_n}x\eqdist \RC_\DD x$ outside the set $\bigcup_{i=1}^n\{R_i>Bn\},$ where $R_i\sim \DD$ is the number of repetitions of $x_i$ when passed through $\TRC_\DD$ or $\RC_\DD$, which by Chebyshev's inequality and a union bound has probability $O(n^{-1})\to 0.$ In Lemma \ref{hp-equal-channels} in the appendix, we show that this implies that the information rates of $\TRC_{\DD_n}$ and $\TRC_\DD $ are the same, as are the information rates of $\RC_{\DD_n}$ and $ \RC_\DD.$\footnote{More precisely, we consider the of channel $\Ch$ which acts on $x\in \{0,1\}^n$ as $\Ch x = \RC_{\CC_n}x,$ (or  $\Ch' x = \TRC_{\DD_n} x$), where $\DD_n$ is the described truncation of $\DD.$ By the ``information rate'' of e.g. $\RC_{\DD_n}$ we mean the information rate of $\Ch$.}

Now assuming $\DD$ is bounded by $Bn>0$ as above, for \eqref{eq:cl1} in Claim~\ref{cl}, we have
\[
I(X;Y|\til{L},\til{R}) = H(X|\til{L},\til{R}) - H(X|Y,\til{L},\til{R}).
\]
We also have $H(X|\til{L},\til{R})\leq H(X)$, and by the chain rule,
\begin{align*}
    H(X|\til{L},\til{R}) &= H(X,\til{L},\til{R}) - H(\til{L},\til{R})\\
    &\geq H(X) - H(\til{L},\til{R}),
\end{align*}
so $H(X|\til{L},\til{R})- H(X)\leq 0$ and $H(X|\til{L},\til{R}) - H(X)\geq -H(\til{L},\til{R}),$ hence $|H(X) - H(X|\til{L},\til{R})|\leq  H(\til{L},\til{R})$ and by an identical derivation also $|H(X|Y)- H(X|Y,\til{L},\til{R})|\leq H(\til{L},\til{R}).$ Hence by the triangle inequality $|I(X;Y) - I(X;Y|\til{L},\til{R})| \leq 2H(\til{L},\til{R})\leq 4\log n = o(n)$ since $(\til{L},\til{R})$ is supported in $[n]^2,$ proving \eqref{eq:cl1}. For \eqref{eq:cl2}, we have
\begin{align*}
    I(X;Y|\til{L},\til{R}) &= I(X_1^{\til{L}}, X_{\til{L} +1}^{\til{R}-1} , X_{\til{R}}^n ;Y_1^L, Y_{L+1}^{R-1}, Y_R^n) \\
    &= I(X_1^{\til{L}}, X_{\til{L} +1}^{\til{R}-1} , X_{\til{R}}^n ;Y_{L+1}^{R-1}) + I(X_1^{\til{L}}, X_{\til{L} +1}^{\til{R}-1} , X_{\til{R}}^n ;Y_1^L, Y_R^n) \\
    &\leq I(X;Y_{L+1}^{R-1}|\til{L},\til{R}) + H(Y_1^L, Y_R^n) \\
    &\leq I(X;Y_{L+1}^{R-1}|\til{L},\til{R}) + H(L, R, |Y|) \\
    &= I(X;Y_{L+1}^{R-1}|\til{L},\til{R}) + o(n),
\end{align*}
where the penultimate inequality is because by definition, $Y_1^L$ and $Y_R^n$ are strings of all zeros, so they are uniquely specified if the length of $Y$ and the indices $L$ and $R$ are given, and the last inequality is because $(L, R, |Y|)$ is supported on $[Bn]^3$. Now by the same argument as in \eqref{eq:cl1}, we again obtain $|I(X;Y_{L+1}^{R-1}|\til{L},\til{R}) - I(X;Y_{L+1}^{R-1})|=o(n),$ and since $Y_{L+1}^{R-1}=\TRIM Y,$ we get $|I(X;Y|\til{L},\til{R}) - I(X;Y')| = o(n),$ where $Y' = \TRC_\DD X.$ This then gives 
\[
\lim_{n\to\infty} \frac{1}{n}\sup_X I(X;Y|\til{L},\til{R}) = \lim_{n\to\infty} \frac{1}{n}\sup_X I(X;Y'),
\]
where $Y=\RC_\DD X$ and $Y'=\TRC_\DD X$, proving \eqref{eq:cl2}.
This finishes the proof of Claim~\ref{cl} and hence of the lemma.

\end{proof}

Next, we show that we may assume an approximately balanced distribution of zeros and ones in all sufficiently long substrings of all codewords in an information-rate-achieving code. The following lemma, though simple, constitutes the substantial improvement in our argument as compared to those of \cite{con-shpilka} or \cite{explicit-codes1}. We remark that this result is much more general than just the setting of repeat channels that we are considering, and in particular applies to all channels to which the Generalized Shannon's Theorem \ref{shannon-thm} applies; nevertheless, for simplicity we state the lemma in the context relevant to our proof.

\begin{prop}\label{zeros-distribution-lemma}
Fix a square-integrable $\DD$-repeat channel $\Ch =\RC_\DD$, or the trimming version $\Ch = \TRC_\DD$, with information rate $\II$. For every $\zeta ,\eps \in (0,1)$ there exists $\gamma \in (0,\frac{1}{2})$ and a family of codes $\CC_n\subseteq \{0,1\}^n$ for $\Ch$ with rate $R\geq \II-\eps$ such that for every $c\in \CC_n$ and $i\in [n-\zeta n],$ we have $ \gamma \zeta n \leq w(x_{i}^{i + \zeta n}) \leq (1-\gamma)\zeta n,$ where $w:\{0,1\}^*\to \N$ is the Hamming weight (number of ones).
\end{prop}
\begin{proof}
The result follows from the fact that in the generalized Shannon's Theorem, we may assume that the process which achieves the information rate is stationary ergodic (see Theorem \ref{shannon-thm}). Even if we're dealing with the trimming version of such a channel, by Lemma \ref{trim-bdc-cap-lemma}, the same statement holds.\footnote{In fact, the statement of Lemma \ref{trim-bdc-cap-lemma} is that the information rates coincide; but by looking at the proof it's clear that we prove the stronger statement that each fixed process $\{X_j\}_{j\geq 1}$ satisfies $\lim_{n\to\infty}\frac{1}{n}I(X;Y) = \lim_{n\to\infty}\frac{1}{n}I(X;Y')$ for $Y=\RC_\DD X$ and $Y'=\TRC_\DD X.$ Hence the information rate of the $\TRC_\DD$ is again attained by the stationary ergodic processes.} Now let $\{X_j\}_{j\geq 0}$ be a stationary ergodic process such that 
\[
\II = \lim_{n\to\infty}\frac{1}{n}I(X_1^n;Y^n),
\]
where $Y^n = \Ch X_1^n$. Let $P := \pr(X_1=1),$ and note that by stationarity we have $P \in (0,1)$ or else $\{X_j\}$ is a trivial process, hence does not achieve the information rate. By Birkhoff's pointwise ergodic theorem, almost surely 
\[
P = \lim_{t\to \infty}\frac{1}{t}\sum_{j=1}^t \1\{X_j=1\} =  \lim_{t\to \infty}\frac{1}{t} w(X_1^t),
\]
so in particular setting $t=\zeta n,$ for any $\delta>0,$ with probability $p_n\to 1,$ we have $(P-\delta)\zeta n\leq  w(X_1^{\zeta n}) \leq (P+\delta)\zeta n.$ Picking $\delta,\gamma$ small enough, we can ensure that $ \gamma\zeta n\leq  w(X_1^{\zeta n}) \leq (1-\gamma)\zeta n$ with probability $p_n$. Now to extend to the substrings, we first look at disjoint consecutive blocks: by stationarity we have $X_{i\zeta n + 1}^{(i+1)\zeta n} \eqdist X_1^{\zeta n}$ for all $i\in [1/\zeta],$ so by a union bound over a constant $1/\zeta$ number of substrings, with probability $\til{p}_n\to 1$ we have $ \gamma \zeta n \leq w(x_{i\zeta n + 1}^{(i+1)\zeta n}) \leq (1-\gamma)\zeta n$ simultaneously for all $i\in [1/\zeta]$. But we note that each substring $x_i^{i+ 3\zeta n}$ fully contains at least one block substring of the form $x_{j\zeta n + 1}^{(j+1)\zeta n};$ hence $\frac{\gamma}{3} 3\zeta n\leq w(x_i^{i+ 3\zeta n})\leq (1-\frac{\gamma}{3})3\zeta n$ for all $i \in [n-\zeta n]$ simultaneously with probability $\til{p}_n.$ Then re-setting $\til{\zeta} = 3\zeta$ and $\til{\gamma} = \gamma/3$ yields the property of the lemma with probability $\til{p}_n.$ Finally we note that we may extract a family of codes $\til{\CC}_n$ of rate $R \geq \lim_{n\to\infty}\frac{1}{n}I(X_1^n;Y^n) - \eps$ from $\{X_j\}$ via sampling, as in the standard proof of Shannon's theorem, and as extended by Dobrushin \cite{generalized-shannons-thm} (see also \cite{cover-thomas}, Theorem 7.7.1, for the argument in the memoryless case, which analogous). Since with high probability this process satisfies the required property, we may discard any codewords from $\til{\CC}_n$ that don't satisfy it to obtain our desired family of codes $\CC_n$ of the same rate. This concludes the proof.
\end{proof}

With these two lemmas, we are ready to prove the correctness of the construction from Section \ref{construction}. 

\begin{proof}[Proof of Theorem \ref{main-thm}]
It remains to show that the decoding algorithm $\Dec$ described in Section \ref{construction} succeeds with high probability, for properly chosen inner code blocklength $m$. There are four potential sources of error in the decoding; the first three pertain to identifying the buffers of zeros, and the fourth to the inner code failures.
\begin{enumerate}
    \item For a given buffer $0^b$ at the sender, less than $\frac{\mu}{2} b =\frac{\mu}{2}\eta m$ zeros survive, so the buffer is not identified during decoding.
    \item All ones in a given inner codeword are deleted, so that two adjacent buffers are incorrectly merged during decoding.
    \item A substring of a received inner word longer than $\frac{\mu}{2}\eta m$ arrives with all zeros (all ones get deleted by the $\BDC$), so that a spurious buffer appears.
    \item For a given correctly identified received inner word, the inner code decoding fails.
\end{enumerate}
We note that error (1) results in the merging of two inner codewords in the decoding process. Since this merged codeword is not the output of the $\TRC_\DD$ with an inner codeword as input, we have no guarantee of a small probability of decoding error of the inner code. We consider the worst-case scenario: we assume the inner code always fails in decoding this string. At the outer code level, this then results in the deletion of two letters, and the insertion of another in the same location, i.e. an edit distance of 3. For error (2), we clearly have a deletion at the outer code level, i.e. an edit distance of 1. For error (3), we again cannot assume the inner code will succeed in decoding these two halves of a received codeword, and hence we assume the worst case scenario: one deletion and two insertions, i.e. edit distance 3. Finally for error (4) we clearly have a substitution at the outer code level, (which is equivalent to a deletion followed by an insertion), i.e. edit distance 2.

Now suppose that each of these errors occurs at most $k\delta / 9$ times. Then the total edit distance is at most $k\delta/9 \cdot (3+1+3+2) = k\delta$. Hence to conclude the proof we must show that each error occurs more than $k\delta/9$ times with vanishing probability, for properly chosen $m$. This then implies that our outer code has to correct from an edit distance more than $k\delta$ with vanishing probability, i.e. the outer code succeeds with probability approaching $1$ as $k\to\infty$ (hence $n\to\infty$).

Error (1) occurs with probability $\pr(|\RC_\DD 0^{m\eta}| < \frac{\mu}{2}\eta m) = O(m^{-1})$ by Chebyshev's inequality. Hence for any $\eta,$ taking $m$ to be a large enough constant we can make this probability less than $\delta/10.$ Since this error can happen independently for each of the $k-1$ buffers, the number of buffers that suffer from error (1) is given by a $Binomial(k-1, p)$ random variable, where $p\leq \delta/10.$ Again by a standard concentration bound, the probability that there are more than $k\delta/9$ errors vanishes as $k\to\infty,$ as desired.

By Proposition~\ref{zeros-distribution-lemma}, each inner codeword has at least $\gamma m$ ones, for some $\gamma>0$ independent of $m.$ Hence error (2) occurs with probability $d^{\gamma m}$. As before, we take $m$ large enough such that $d^{\gamma m}<\delta/10,$ and then by concentration of measure as $k\to \infty,$ the probability of having more than $k\delta/9$ errors vanishes.

We now consider error (3). Consider the event that we receive a string $s$ of all zeros with $|s|\geq \frac{\mu}{2}\eta m$ as part of the output of the channel for a codeword $x\in \CC_{in}$ as input. This implies one of two things: ($a$) that some substring $\til{s}$ of length $>\frac{1}{4}\eta m$ of the input had all its one bits deleted and gave rise to $s$, or ($b$) that some substring $\til{s}$ of length $\leq \frac{1}{4}\eta m$ at the input gave rise to \emph{any} string of length $\geq \frac{1}{2}\mu \eta m$ at the output. We analyze each case separately. For ($a$), by Proposition~\ref{zeros-distribution-lemma}, choosing $\zeta = \frac{1}{2}\mu\eta$, we must have $w(\til{s})\geq \gamma \zeta m.$ But then the probability that such a substring $s$, say at the beginning of the received word, exists in the first place is less than $d^{\gamma \zeta m},$ and by a union bound the probability that any such substring exists is less than $O(1) \cdot d^{\gamma \zeta m}$ (since the received word has length $\leq m,$ and hence we can discretize it into $O(1)$ substrings of size $\geq \frac{\mu}{2}\eta m$) which can be made less than $\delta/20$ for $m$ chosen large enough. For ($b$), note that a substring of length $\leq \frac{1}{4}\eta m$ at the input giving length $\geq \frac{1}{2}\mu\eta m$ at the output implies that there's a substring of length \emph{exactly} $\frac{1}{4}\eta m$ giving an output of length $\geq \frac{1}{2}\mu\eta m$ (since a bigger input can only give a bigger output). But if $Z = X_1+\dots+X_t$, for $t = \frac{1}{4}\eta m$ and $X_j\sim \DD,$ the probability of this happening is 
\begin{align*}
    \pr(Z\geq \frac{1}{2}\mu \eta m) &\leq \pr(|Z - \E Z| \geq \frac{1}{4}\mu \eta m) \\
    &\leq \frac{t\sigma^2}{(\frac{1}{4}\mu \eta m)^2} \\
    &= \frac{\sigma^2}{\frac{1}{4}\mu \eta m} \\
    &= O(m^{-1})
\end{align*}
by Chebyshev's inequality. Again by a union bound over $O(1)$ possible initial substrings $s,$ making $m$ large enough we can make this $\leq \delta/20.$ Hence, the probability of error (3) is $\leq \delta/20 + \delta/20 = \delta/10,$ and by concentration of measure, more than $\delta/9$ errors occur with vanishing probability.

Error (4) occurs with probability going to zero as $m$ grows by assumption of the inner code being sound for the $\TRC_\DD$. For $m$ large enough this probability is less than $\delta/10,$ and by the same argument as above as $k\to\infty$ we get $k\delta/9$ errors with vanishing probability. This concludes the proof.

Finally, the error probability is $e^{-\Omega(n)}$ because, as was mentioned, the frequency of each error type (1-4) is a $Binomial(t, p)$ random variable with $t=k-1$ or $t=k$ and $p\leq \delta/10.$ Hence by a standard Chernoff bound, and union bounding over errors (1-4), we obtain the desired $e^{-\Omega(n)}$ probability of edit distance greater than $k\delta/9$, i.e. a $e^{-\Omega(n)}$ probability of failure. This concludes the proof.

\end{proof}

\section{Extensions to More General Channels}\label{extensions}\label{sec:general}

In this section, we explore how our work generalizes beyond square-integrable repeat channels, and sketch how it can be extended to \emph{biased square-integrable Dobrushin channels}, which we introduce below.

A natural extension of the repeat channel model is given by the class of channels admitted by the conditions of the Generalized Shannon's Theorem \ref{shannon-thm}. These are the \emph{Dobrushin channels}, which appear in the work of Pfister and Tal \cite{polar2}.
\begin{defn}
Fix two probability distributions $\DD_0$ and $\DD_1$ over $\{0,1\}^*.$ A \emph{$(\DD_0,\DD_1)$-Dobrushin channel} $\DC_{\DD_0,\DD_1}$ (we often just write $\DC$) acts independently on each input bit as $(\DC 0) \sim \DD_0$ and $(\DC 1)\sim \DD_1,$ and concatenates the outputs. We say that $\DC$ is \emph{square-integrable} if for $Y_i\sim \DD_i$ we have $\E|Y_i|^2<\infty$ for $i=0,1.$
\end{defn}
We note that repeat channels correspond to the special case of Dobrushin channels where $\DD_0$ is supported on all-zero strings, $\DD_1$ is supported on all-one strings, and their induced distributions of string lengths coincide. Moreover in that case, both notions of square-integrability agree. We also note that many natural notions of insertion, substitution, and deletion errors can be described as Dobrushin channels.

Can we extend our construction from Section~\ref{sec:main-result} to general square-integrable Dobrushin channels? A moment of thought reveals that many Dobrushin channels have a capacity of zero, and hence our construction fails to make sense (note for example that in our proof of Proposition~\ref{zeros-distribution-lemma}, we assume that the information rate of the channel is non-zero). Even if the capacity is non-zero, for general distributions $\DD_0,\DD_1$, it's unclear how to distinguish the buffers of zeros from the inner codewords at the receiver. A simple condition on the distributions $\DD_0,\DD_1$ was considered by \cite{polar2} that allows one to reliably identify the buffers at the receiver; we call these the \emph{biased} Dobrushin channels.
\begin{defn}\label{biased-dobrushin}
A Dobrushin channel $\DC_{\DD_0,\DD_1}$ is \emph{biased} if for $Y_0\sim \DD_0,Y_1\sim \DD_1$ we have $\E |Y_0| =\E|Y_1|<\infty,$ and $\E [w(Y_0)]<\frac{1}{2}\E|Y_0|$ and $\E [w(Y_1)]>\frac{1}{2}\E|Y_1|$, where $w:\{0,1\}^*\to \N$ is the Hamming weight.
\end{defn}

Pfister and Tal \cite{polar2} proved that for a restricted class of Dobrushin channels, allowing only for deletions, substitutions, and one-bit insertions, there is a polar codes construction that can achieve the capacity. As they mention, their construction can also be generalized to the biased Dobrushin channels, if some mild regularity conditions are imposed on the $\DD_i.$ Their encoding algorithm runs in linear time, and their decoding algorithm runs in time $O(n^{1+3\nu})$ with failure probability $e^{-\Omega(n^{\nu'} )},$ for any $0<\nu'<\nu<\frac{1}{3}.$ We now give a proof sketch of how our construction may be extended to produce near capacity-achieving codes for biased square-integrable Dobrushin channels, with a slightly improved probability of failure $e^{-\Omega(n)}$ and a slightly improved decoder runtime of $O(n\poly(\log n)).$

\begin{thm}\label{extended-main-thm}
    Fix a biased square-integrable Dobrushin channel $\DC.$ For every $\eps>0$, there exists a sound code $\CC$ with rate $R$ for the $\DC$ with $R\geq Cap(\DC) - \eps$ and linear and quasi-linear time encoding and decoding algorithms, respectively. Moreover, the decoder has probability of failure $e^{-\Omega(n)}.$
\end{thm}
We now outline how to extend our arguments from Section~\ref{sec:main-result} to this more general setting. Throughout, we will refer to the expected fraction of ones in the output of a long sequence of zeros passed through the channel, which we denote by the parameter
\[
f = \frac{\E [w(Y_0)]}{\E|Y_0|},
\]
where $Y_0$ is as in Definition~\ref{biased-dobrushin}; by that definition we have $f<\frac{1}{2}.$ We follow a somewhat different order to the one of Section~\ref{sec:main-result}. First, we outline the necessary modifications to our construction from Section~\ref{construction}. Second, we prove the low probability of decoding failure (the analog of which for repeat channels was proved in the proof of Theorem~\ref{main-thm} in Section~\ref{main-proof}). Finally, we return to the capacity of the trimming version of the channel, needed to obtain our rate guarantees.

\paragraph{Construction.} We essentially leave our construction from Section~\ref{construction} unchanged; the only difference to address is the identification of the zero buffers by the decoder. We declare any contiguous block of $\nu \eta m$ bits with a fraction of ones less than $f+\kappa$ to be part of a buffer, where $\nu$ and $\kappa$ are parameters to be chosen later. To implement this efficiently, one can iterate a running window of size $\nu \eta m$ through the received codeword, and as soon as the fraction of ones goes below $f +\kappa,$ we say we have hit a buffer; as soon as that fraction goes above $f +\kappa,$ we declare that the current buffer has ended. As in Section~\ref{construction}, we remove these buffers, and process the inner codewords in between as before.

\paragraph{Correctness.} The sources of error are similar to those of the proof of Theorem~\ref{main-thm}. Only two additional claims are needed: (1) the probability that any substring of length $\nu\eta m$ of a received \emph{buffer} has a frequency of ones \emph{higher} than $f+\kappa$ vanishes as $m\to\infty$, and (2) the probability that any substring of length $\nu\eta m$ of a received \emph{codeword} has a frequency of ones \emph{lower} than $f+\kappa$ vanishes as $m\to\infty.$ As we will see, the intermediary cases where such a substring falls half in a buffer and half in a codeword will not be necessary to analyze. In the Appendix \ref{app-lem-2} we sketch a proof of claim (2);  (1) follows by an easier version of the same argument, with the caveat that we must take $\nu$ small to ensure that the buffer is identified with high probability. 

\paragraph{Trimming Channel Information Rate.} We need an analog of Lemma~\ref{trim-bdc-cap-lemma} for this more general setting. Given the discussion of the previous paragraph, the appropriate channel is one which trims a random number of bits off each end of the codeword bounded by $\nu \eta m.$ Along these lines we define the \emph{trimming Dobrushin channels}.
\begin{defn}
Given an integer $n$ and two probability distributions $\TT_\ell,\TT_r$ (allowed to depend on $n$) over the natural numbers, let $\TRIM_{\TT_\ell,\TT_r}$ be the channel which acts on $x\in \{0,1\}^n$ as $\TRIM_{\TT_\ell,\TT_r} x = x_{t_\ell}^{n-t_r}$ where $t_\ell \sim \TT_\ell,t_r\sim \TT_r$ are independent (with $\TRIM_\TT x$ the empty string if $n-t_2<t_1$). Given a Dobrushin channel $\DC = \DC_{\DD_0,\DD_1},$ we define the \emph{$(\TT_\ell,\TT_r)$-trimming $(\DD_0,\DD_1)$-Dobrushin channel} by the concatenation $\TDC_{\DD_0,\DD_1,\TT_\ell,\TT_r}=\TDC_{\TT_\ell,\TT_r}=\TRIM_{\TT_\ell,\TT_r}\circ \DC.$
\end{defn}
\begin{lem}\label{trim-dobrushin-cap-lemma}
Let $\DC$ be a square-integrable Dobrushin channel, and let $\TDC_{\TT_\ell,\TT_r}$ be the trimming version with $\TT_\ell,\TT_r$ each bounded by $\nu \eta m$ for inputs of length $m.$ Then the information rate for the $\TDC_{\TT_\ell,\TT_r}$ is at least $Cap(\DC) - 2\nu\eta.$
\end{lem}
\begin{proof}[Proof Sketch]
See the Appendix \ref{app:trim-dobrushin-cap-lemma}.
\end{proof}

\bigskip

We finally conclude by verifying the promised rate, runtime, and probability of decoding failure guarantees.
\begin{proof}[Proof Sketch of Theorem \ref{extended-main-thm}]
The rate is as in Section~\ref{construction}, except we now have $n=km/(Cap(\DC) - \nu\eta - \psi(\eps,\delta, \nu, k,m))$. Again, if we take $\eps,\delta,\eta,\nu\to 0$ and $k,m\to\infty$, the rate converges to $Cap(\DC),$ and hence taking $\eps,\delta,\eta,\nu$ small enough and $m$ large enough we can make the rate get arbitrarily close to $Cap(\DC).$

As for the runtime, the encoding is identical to the one in Section~\ref{construction} and hence is also $O(n)$. For the decoding, the only modification we have made is the identification of the buffers ---given our discussion above, this is also $O(n),$ and hence the decoding runtime is still quasi-linear in $n,$ as desired.

As for the probability of decoding error, the same argument as in the proof of Theorem~\ref{main-thm} applies: the frequency of errors of each kind are binomial random variables, and hence by a Chernoff bound the probability of decodign failure is $e^{-\Omega(n)}.$
\end{proof}

\section*{Acknowledgements}
We thank Ido Tal for pointing out an error in our description of \cite{polar,polar2} in an earlier version of this manuscript.

\printbibliography
\nocite{*}

\appendix

\input{app.tex}

\end{document}

%% file: app.tex
\section{Omitted Proofs}\label{app}
\begin{lem}\label{hp-equal-channels}
Let $\Ch_1, \Ch_2:\Omega\times \{0,1\}^*\to\{0,1\}^*$ be two channels. Suppose there exists a sequence $p_n\geq0$ with $p_n\to 1$ as $n\to \infty$ such that for every $x\in \{0,1\}^n,$ there is a set $A_x\subset \Omega$ such that, conditioned on $A_x,$ the distributions of $\Ch_1 x$ and $\Ch_2 x$ are the same, and $\pr(A_x)\geq p_n.$ Then the information rates of $\Ch_1$ and $\Ch_2$ are the same.
\end{lem}
\begin{proof}
Letting $Y_1 = \Ch_1 X$ and $Y_2=\Ch_2 X,$ we have by the chain rule 
\[
I(X;Y_i)\leq I(X,\1_{A_X};Y_i) = I(\1_{A_X};Y_i) + I(X;Y_i|\1_{A_X}) \leq I(X;Y_i|\1_{A_X}) + \log 2
\]
for $i=1,2,$ and similarly $I(X,\1_{A_X};Y_i) = I(X;Y_i) + I(\1_{A_X};Y_i|X)$ so 
\[
I(X;Y_i) = I(X;Y_i|\1_{A_X}) + I(\1_{A_X};Y_i) - I(\1_{A_X};Y_i|X)\geq I(X;Y_i|\1_{A_X}) - 2\log 2,
\]
so $|I(X;Y_i) - I(X;Y_i|\1_{A_X})|\leq 2\log 2 = o(n)$, and to show that the information rates are the same, it suffices to prove $\lim_{n\to \infty}\frac{1}{n}I(X;Y_1|\1_{A_X})=\lim_{n\to \infty}\frac{1}{n}I(X;Y_2|\1_{A_X}).$ But we have
\begin{align*}
    \lim_{n\to \infty}\frac{1}{n}I(X;Y_1|\1_{A_X}) &= \lim_{n\to \infty}\frac{1}{n}[\pr(A_X)I(X;Y_1|A_X) + \pr(A_X^c)I(X;Y_1|A_X^c)] \\
    &= \lim_{n\to \infty}\frac{1}{n}\pr(A_X)I(X;Y_1|A_X) &&(*) \\
    &= \lim_{n\to \infty}\frac{1}{n}\pr(A_X)I(X;Y_2|A_X) \\
    &= \lim_{n\to \infty}\frac{1}{n}[\pr(A_X)I(X;Y_2|A_X) + \pr(A_X^c)I(X;Y_2|A_X^c)] &&(*)\\
    &= \lim_{n\to \infty}\frac{1}{n}I(X;Y_2|\1_{A_X}) 
\end{align*}
where $(*)$ is since $0\leq I(X;Y_i|A_X^c)\leq n$ for $i=1,2$ and $\pr(A_X^c)\to 0$, so the term $\pr(A_X^c)I(X;Y_i|A_X^c)$ doesn't contribute to the limit. This concludes the proof.
\end{proof}

\begin{lem}\label{app-lem-2}
Let $c$ be a codeword of length $m$ from the near capacity-achieving code for a biased square-integrable Dobrushin channel $\DC=\DC_{\DD_0,\DD_1}$ of Lemma~\ref{zeros-distribution-lemma}. The probability that any length $\nu \eta m$ substring from $\DC c$ has a fraction of zeros less than $f + \kappa$, where $f = \E[w(Y_0)]/\E|Y_0|$ and $Y_0\sim \DD_0$, vanishes as $m\to\infty,$ for appropriately chosen $\zeta > 0$ in \ref{zeros-distribution-lemma} and for any $\nu,\eta>0$ and all $\kappa>0$ small enough.
\end{lem}
\begin{proof}[Proof Sketch]
We let $Y^1,Y^2,\dots,Y^t$ be the sequence of all outputs of the $\DC$ on individual bits of $c$ that are fully contained in the first $\nu\eta m$ bits of $\DC c.$ We note then that the window of the first $\nu\eta m$ bits in $\DC c$ consists of the concatenation $Y^1Y^2\dots Y^t$, followed by potentially a few bits from the output of the $(t+1)$st bit in $c$, which did not fall entirely within the first $\nu\eta m$ bits; if we were considering a window of bits in the middle of the codeword $c,$ we would have such ``leftover bits'' on either side of the concatenation $Y^1\dots Y^t.$ It's easy to see that in either case these leftover bits will be $o(m)$ in number with high probability, and so we will ignore them for the rest of the derivation, by assuming that we have exactly $|Y^1\dots Y^t|=\nu\eta m.$ Moreover, by a union bound over a constant $1/\nu\eta$ number of blocks (and putting together blocks to form all substrings as in the argument of Lemma~\ref{zeros-distribution-lemma}), proving the claim for the $Y^1,Y^2,\dots,Y^t$ at the beginning of the received codeword suffices. With these assumption, we want to show that
\begin{equation}\label{lemA2:1}
    \pr\left(\frac{\sum_{j=1}^t w(Y^j)}{\sum_{j=1}^t |Y_j|} < f+\kappa\right) = 
\pr\left((f+\kappa)\sum_{j=1}^t |Y_j| - \sum_{j=1}^t w(Y^j) > 0\right)
\end{equation}
vanishes as $m\to\infty.$ But, noting that $t$ is a random quantity, we have
\begin{align*}
    \E\left[ (f+\kappa)\sum_{j=1}^t |Y_j| - \sum_{j=1}^t w(Y^j)\right]
    &= \E\left[\E\left[ (f+\kappa)\sum_{j=1}^t |Y_j| - \sum_{j=1}^t w(Y^j) \;\Bigg| \;t\;\right]\right] \\
    &= (f+\kappa)(\E t)\E|Y_0| - \E \left[\sum_{j=1}^t \E[w(Y^j)]\right].
\end{align*}
Now note that by Lemma~\ref{zeros-distribution-lemma}, we can ensure that each substring of length $\geq \nu \eta m /(2\E t)$ from $c$ will have at least $\gamma \nu \eta m /(2\E [t]\E |Y_0|)$ ones. But then so long as $t\geq \frac{1}{2}\E t$ (which clearly happens with probability $\to 1$ as $m\to\infty$) we will have
\begin{align*}
    (f+\kappa)(\E t)\E|Y_0| - \E \left[\sum_{j=1}^t \E[w(Y^j)]\right] &\leq (f+\kappa)(\E t)\E|Y_0| - \left(\gamma(\E t) \E[w(Y_1)]  + (1-\gamma) (\E t)  \E[w(Y_0)] \right) \\
    &\leq (\E t)\E|Y_0| (f+\kappa - (\gamma/2 + (1-\gamma) f)),
\end{align*}
and since $f<\frac{1}{2}$ and $\gamma>0,$ choosing $\kappa$ small enough makes this equal to $-C\E t$ for $C>0$ a constant, and then clearly $-C\E t = -\Omega(m).$ Hence by concentration of measure (using that variances are $o(m^2)$ by standard arguments), as $m\to \infty$, \eqref{lemA2:1} vanishes, as desired.
\end{proof}

\begin{lem}[Lemma~\ref{trim-dobrushin-cap-lemma} in the main text.] \label{app:trim-dobrushin-cap-lemma}
Let $\DC$ be a square-integrable Dobrushin channel, and let $\TDC_{\TT_\ell,\TT_r}$ be the trimming version with $\TT_\ell,\TT_r$ each bounded by $\nu \eta m$ for inputs of length $m.$ Then the information rate for the $\TDC_{\TT_\ell,\TT_r}$ is at least $Cap(\DC) - 2\nu\eta.$
\end{lem}
\begin{proof}[Proof Sketch.]
Following the proof of Lemma~\ref{trim-bdc-cap-lemma}, we let $T_\ell\sim \TT_\ell, T_r\sim \TT_r$ be the number of bits that get trimmed off the left and right ends of the output $Y=\DC X$ ($X$ supported on $\{0,1\}^m$), respectively, and let $L=T_\ell, R=m-T_r.$ We prove the same Claim~\ref{cl} with the modified second part 
\[\label{eq:cl2'}
\lim_{m\to\infty}\frac{1}{m}\sup_X I(X;Y|\til{L},\til{R}) \leq \lim_{m\to\infty} \frac{1}{m}\sup_X I(X,Y') + 2\nu\eta, \tag{2'}
\]
where $Y' = \TDC_{\TT_\ell,\TT_r} X.$ The proof of \eqref{eq:cl1} is identical. For \eqref{eq:cl2'}, by the same argument we arrive at the inequality
\begin{align*}
    I(X;Y|\til{L},\til{R}) &\leq I(X; Y_{L+1}^{R-1}|\til{L},\til{R}) + H(Y_1^L, Y_R^m),
\end{align*}
but now all we can claim is $H(Y_1^L, Y_R^m)\leq \log ((2^{m\nu\eta})^2) =2 m\nu\eta.$ The rest of the argument proceeds exactly as before and gives \eqref{eq:cl2'}, and hence the lemma, as desired.
\end{proof}